\documentclass[aps,prd,preprint,nofootinbib]{revtex4-2}
\usepackage[pdftex]{graphicx,xcolor,hyperref}
\usepackage{amsmath,amssymb,bm,amsthm,mathtools}
\usepackage{amsfonts}
\usepackage{ascmac}
\usepackage{array}
\usepackage{slashed}
\usepackage{booktabs} 
\usepackage{tikz}

\newtheorem{proposition}{Proposition}
\newtheorem{lemma}{Lemma}

\DeclareMathOperator{\spec}{Spec}

\date{\today}
\begin{document}

\title{Symmetric Mass Generation for Domain-Wall Fermions}
\author{Sho Araki}\email{araki@het.phys.sci.osaka-u.ac.jp}
\author{Hidenori Fukaya}
\email{hfukaya@het.phys.sci.osaka-u.ac.jp}
\author{Tetsuya Onogi}
\email{onogi@het.phys.sci.osaka-u.ac.jp}
\author{Satoshi Yamaguchi}
\email{yamaguch@het.phys.sci.osaka-u.ac.jp}
\affiliation{Department of physics, The University of Osaka, Toyonaka 560-0043, Japan}
\preprint{OU-HET-1322}
\begin{abstract}
We report the first numerical examination of symmetric mass
generation (SMG) for domain-wall fermions. Our simulation on a
two-dimensional Euclidean lattice computes the path-integral
of the Fidkowski–Kitaev Majorana chain with two
boundaries, corresponding to the two (physical and mirror) walls.
We demonstrate by evaluating the lattice Dirac operator eigenvalue spectrum
that the local four-Fermi interaction selectively gaps the edge fermions on the mirror wall
while leaving the physical wall intact. We further identify a
characteristic signature of SMG in the path-integral: the gap
of the Dirac operator opens along the imaginary axis of the
complex spectrum, in contrast to the standard mass shift in the real direction.
The conventional hybrid Monte Carlo algorithm remains practical with reweighting the complex phase of the Pfaffian,
which is found to be remarkably small throughout the simulation.
These results provide a practical nonperturbative framework for 
interaction-induced mirror-fermion decoupling of lattice domain-wall fermions.
\end{abstract}
\maketitle

\section{introduction} \label{sec:intro}

Domain-wall fermions \cite{Kaplan:1992bt, Shamir:1993zy}
(as well as the overlap fermions \cite{Neuberger:1997fp},
which are obtained by integrating out the bulk part of domain-wall fermions)
have played
a crucial role in understanding chiral symmetry \cite{Ginsparg:1981bj,Luscher:1998pqa} in lattice gauge theory.
However, their application has been limited to vectorlike gauge theories
and constructing lattice chiral gauge theories remains a challenge,
except for a few successful examples
where the gauge group is (or contains) $U(1)$ \cite{Luscher:1998du,Kikukawa:2000kd}.

\newpage
This difficulty comes from the fact that we need to have,
at least, two domain-walls on a standard square lattice\footnote{
A single domain-wall is possible when it is allowed to be curved.
See Refs.~\cite{Aoki:2022cwg, Aoki:2022aez, Kaplan:2023pxd, Kaplan:2023pvd, Aoki:2024bwx, Clancy:2024bjb}
for details.
},
where both of the left- and right-handed fermions appear in pairs.
Consider a set of Weyl fermions localized on one of the two domain walls, which we call the physical wall.
Then we have the same set
but with the opposite chirality sitting on another wall
which we denote the mirror wall.
The system is, thus, automatically vectorlike, rather than chiral.
In order to solve this so-called mirror fermion problem\cite{Eichten:1985ft,Creutz:1996xc},
we need to somehow decouple the fermions on the mirror wall from the theory.

Symmetric mass generation (SMG) \cite{Fidkowski:2009dba,Fidkowski:2010jmn, Wen:2013ppa, Wang:2013yta,You:2014oaa,Ayyar:2014eua,BenTov:2014eea,You:2014vea,Catterall:2015zua,Ayyar:2015lrd, He:2016sbs,Ayyar:2016lxq,You:2017ltx,Kikukawa:2017gvk,Wang:2018jkc,Razamat:2020kyf,Tong:2021phe,Butt:2021koj,Zeng:2022grc,Wang:2022ucy,Liu:2023msa,Butt:2024kxi,Golterman:2025boq,Maiti:2026log,Li:2026gvn, Hasenfratz:2026orf,butt2026searchingsymmetricmassgeneration}
has attracted considerable attention in condensed matter and particle physics.
When the fermions are free of gauge and 'tHooft anomalies \cite{tHooft:1979rat},
SMG may give a mass to these fermions by multi-point fermion interactions
without breaking the symmetry that prohibits the standard
bi-linear mass term.
If SMG can be implemented for domain-wall fermions, namely by selectively gapping the mirror-wall fermions while leaving the physical wall intact, it would provide a solution to the mirror-fermion problem.

Application of SMG to domain-wall fermions has been proposed in the literature
\cite{Wen:2013ppa, Wang:2013yta, You:2014oaa, You:2014vea, Wang:2018jkc}
but no nonperturbative numerical study has been
reported\footnote{
Kikukawa \cite{Kikukawa:2017gvk} reported mirror fermion decoupling
using the overlap Dirac operator in numerical simulations of two-dimensional chiral fermion models,
rather than domain-wall fermions. 
}\footnote{Zeng et al. \cite{Zeng:2022grc} simulated a similar set up with two boundaries
and were successful in gapping out the mirror fermions on one edge.
However, the two edges are separated by only one lattice spacing
and the symmetry is still exact, so that we consider the formulation in \cite{Zeng:2022grc}
as a kind of staggered fermions.
In contrast, the chiral symmetry of domain-wall fermions improves
exponentially as the separation between the two walls increases.}.
Previous numerical studies employed lattice fermion formulations with multiple low-energy poles,
including honeycomb-lattice models \cite{He:2016sbs,Liu:2023msa} and staggered fermions \cite{Ayyar:2014eua,Catterall:2015zua,Ayyar:2015lrd,Ayyar:2016lxq, Butt:2021koj,Butt:2024kxi}.

In this work, we simulate eight-flavor
Majorana domain-wall fermions on two-dimensional lattices.
We numerically verify by the lattice Dirac operator spectrum
that the SMG by a set of four-Fermi interaction
works only on the mirror wall, while keeping the massless fermion Dirac spectrum on
the physical wall unchanged.
Our simulation corresponds to the Euclidean path integral of
the Fidkowski-Kitaev(FK) Majorana chain \cite{Fidkowski:2009dba,Fidkowski:2010jmn}.
Although this model is not a chiral gauge theory,
the edge-localized modes, which are eigenstates of $\gamma^1$ related to the reflection symmetry
instead of $\gamma^5$,
have exactly the same mechanism that they cannot have the mass term.
The aim of this work is to establish a solid benchmark of SMG in the path-integral formalism,
as a solution to the mirror fermion problem of domain-wall fermions.

The remainder of this paper is organized as follows.
We first describe the eight-flavor Majorana fermion theory converting
the four-Fermi interaction into the Yukawa interaction
with auxiliary scalar fields.
Then we show that the standard hybrid Monte Carlo algorithm
with Pfaffian phase reweighting
is practical, as the complex phase is found to be remarkably small.
We study the lattice Dirac operator spectrum without domain walls
and identify a characteristic
signature of SMG in this path-integral formalism: the spectral gap opens along the imaginary
axis rather than through
a conventional real mass shift. Finally, we introduce domain walls and apply the four-fermion
interaction only near the mirror wall, demonstrating that the mirror fermions are
gapped while the physical-wall spectrum remains gapless.

\section{Two-dimensional eight-flavor Majorana fermions}

We consider the FK model \cite{Fidkowski:2009dba}\footnote{
The original FK Hamiltonian (in its continuum limit)
is recovered
by the Legendre transformation of the Lagrangian above with respect to
$\chi^i = (\eta^i, \bar{\eta}^i)^T$
(See Appendix~\ref{app:originalFK} for the details).
} with eight-flavor Majorana fermions \cite{Tong:2019bbk} 
on a two-dimensional Euclidean flat torus whose action is given by
\begin{equation}
  \label{eq:action}
  S =\int d^2x \left[\frac{1}{2} \chi_i^T C(\gamma^\mu \partial_\mu +m)\chi^i
    -\frac{g_S^2}{2}\left(\frac{1}{2}\chi^T_i C\chi^i \right)^2
    -\frac{g_P^2}{2}\sum_{\alpha=1}^7\left(\frac{1}{2}\chi^T_i C i \gamma^3 [\Gamma^\alpha]^i_j \chi^j \right)^2\right].
\end{equation}
Here $i,j=1,\cdots 8$ are flavor indices,
\begin{align}
  \gamma^1=\left(
  \begin{array}{cc}
    0 & 1\\
    1 & 0\\
  \end{array}
  \right),\;\;\;
    \gamma^2=\left(
  \begin{array}{cc}
    0 & -i\\
    i & 0\\
  \end{array}
  \right),\;\;\;
   \gamma^3=\left(
  \begin{array}{cc}
    1 & 0\\
    0 & -1\\
  \end{array}
  \right),\;\;\;
  C=i\gamma^2,
\end{align}
which act on the two-component  spinor $\chi^i$,
and the seven $8\times 8$ flavor matrices (following the notation by \cite{You:2014vea}),
are defined by
$
 \{\Gamma^\alpha\} =\left\{\sigma^{123},\sigma^{203},\sigma^{323},\sigma^{211},\sigma^{021},\sigma^{231},\sigma^{002}\right\},
$
where $\sigma^{abc}=\sigma^a\otimes \sigma^b\otimes \sigma^c$ with the Pauli matrices $\sigma^{a=1,2,3}$ and $\sigma^0=1$.
Note that $\Gamma^\alpha$ are all pure imaginary anti-symmetric matrices,
which gives the theory a global $SO(7)$ symmetry.

Instead of the original time-reversal symmetry in FK,
we have the reflection symmetry in the Euclidean time direction \cite{Shiozaki:2017ive}:
the action $S$ is invariant under 
$
\chi(x_1,x_2) \to -\gamma^3\gamma^2 \chi(x_1,-x_2),
$
which is an element of $Pin^-(2)$ group denoted by $R_t$.
The Majorana masses or any bilinear operators
$\eta^i\eta^j$ or $\bar{\eta}^i\bar{\eta}^j$ are not allowed with this symmetry,
while the Dirac mass $m \chi^TC \chi$ is allowed.
In the massless limit, $m\to 0$,
$S$ is also
invariant under 
$
\chi(x_1,x_2) \to \gamma^3 \chi(x_1,x_2),
$ 
which is a (discrete) chiral symmetry.
The four-Fermi interaction 
respects both symmetries.
In fact, we would like to examine the SMG in two ways:
1) The chiral SMG: we scan the uniform mass parameter $m$ without domain-walls
to examine continuity between the topological sector 0 and 8 without closing the mass gap.
This provides a nonperturbative test of the $\mathbb{Z}_8$ classification of symmetry-protected topological phases.
2) The reflection SMG : with two domain walls, we examine 
whether the edge-localized fermions on the mirror wall only can be gapped out by the same four-Fermi interactions. 

In this work, we regularize the Majorana fermion theory
on a square lattice with the lattice spacing $a$,
employing the massive Wilson Dirac operator $D_W+m$
with and without domain-walls\footnote{
In our previous work \cite{Araki:2025xly}, we have confirmed
in the free Wilson fermion theory with $g_P^2=0$
that the $\mathbb{Z}_8$ structure of the Pfaffian phase
(for a single flavor fermion)
is well reproduced on the lattice.
}.
We set $g_S^2=0$, which is irrelevant for SMG.
According to \cite{Fidkowski:2009dba}, we expect that SMG works for $g_P^2>0$ and does not for $g_P^2\le 0$.
We convert the four-Fermi term into
the Yukawa interaction with the auxiliary field $\pi^\alpha(x)$ on each site $x=(x_1,x_2)$
to integrate out the Majorana fields and obtain a path integral over $\pi^\alpha(x)$:
\begin{align}
  Z =&\int \prod_{x,\alpha} d\pi^\alpha(x) {\rm Pf}\left[C\left(D_W + m + \sqrt{g_P^2}\sum_{\alpha=1}^7 \pi^\alpha(x) i \gamma_3 \Gamma^\alpha\right)\right]
  \exp\left[ - \frac{1}{2}\sum_x \sum_{\alpha=1}^7 (\pi^\alpha(x))^2   \right], 
\end{align}
where ${\rm Pf}[\cdots]$ denotes the Pfaffian of the eight flavor Dirac operator.
In this expression, the SMG looks like a Higgs mechanism,
but it is different in that
the vacuum expectation value of $\pi^\alpha$ must be zero,
and there is no massless (Nambu-Goldstone) particle in the system.
Note that the $R_t$ symmetry is exact in our lattice formulation.


\section{Lattice simulation setup}

In our lattice simulations, we investigate the mass gap in terms of
the spectrum of the Dirac operator (including the Yukawa term):
\begin{equation}
 D_Y:= D_W + m + \sqrt{g_P^2}\sum_{\alpha=1}^7 \pi^\alpha(x) i \gamma_3 \Gamma^\alpha,
\end{equation}  
as well as that of $D_Y^\dagger D_Y$.
The gap in $D_Y^\dagger D_Y$ guarantees that any fermion
correlator exponentially decays at sufficiently long distances,
which is discussed in Appendix~\ref{app:gap}.

We fix the spatial lattice size to 12
and vary the temporal size in the range $[8,32]$, depending on the parameters.
We impose the periodic boundary conditions in both of
$x_1=x$ and $x_2=t$ directions. 
When we switch off the domain-wall, we simulate with uniform mass values at $m=0.5, 0, -1.5,-2$
and put $\pi^\alpha(x)$ field on every site.
For the domain-wall fermion case, we put $m=-1$ for $4\le x\le 9$, and $+1$, otherwise.
We put $\pi^\alpha(x)$ field only on the $x=8,9,10,11$ slices near the mirror wall.
For the Yukawa coupling, we simulate different values of $g_P^2$ in the range $[-0.05,0.25]$.

We employ the hybrid Monte Carlo (HMC) algorithm ignoring the phase of the Pfaffian:
$|{\rm Pf}[CD_Y]|=|\det(D_Y^\dagger D_Y)|^{1/4}=\exp[{\rm Tr}\ln (D_Y^\dagger D_Y)/4]$.
For our simulations, we partly use JuliaQCD package \cite{Nagai:2024yaf} with some modifications.
We do not use the pseudo fermions but directly
compute the force from the fermion action,
 as well as
 the Pfaffian itself for reweighting the phase,
 in a brute force way.
 The Pfaffian phase fluctuates very mildly around zero:
 its standard deviation is
 at most $\sqrt{\langle \Delta \theta^2\rangle}=$0.015(3) radian
 at $g_P^2=0.1$ on our simulated lattices.
We confirm that the phase effect is negligible 
compared to the statistical errors.
We simulate 1000--3000 trajectories for each ensemble and
sample configurations per 10 trajectories after
discarding 200--1000 trajectories for thermalization.
The statistical errors are estimated by the jackknife resampling method.

We note here that the scalar one-point function $\langle \pi^\alpha(x)\rangle$
is consistent with zero for all simulated ensembles and there is no
sign of spontaneous breaking of symmetries.

\section{Numerical results}

\subsection{Chiral SMG without domain-walls}

First, we examine the chiral SMG by simulations
with uniform mass $m$ without domain-walls.
In continuum theory, the chiral symmetry is manifest at $m=0$,
while there is an additive mass renormalization correction on a lattice.
Therefore, we scan the mass in a range $-2\leq m \leq 0.5$ to examine the chiral SMG.

For the free case $g_P^2=0.0$, the eigenvalue spectrum of $D_Y$ with $m=0$
is plotted on the left panel of Fig.~\ref{fig:Dspectrum}.
Note that adding nonzero mass $m$ shifts
the whole spectrum along the real axis, setting the new origin at $-m$. 
The region $0<-m<2$, whose $\mathbb{Z}$ topological number is $k=+8$
and the region $2<-m<4$ with $k=-8$ are clearly separated from the trivial phase $k=0$ for $-m<0, 4<-m$,
by the eigenvalues at $0$ with eight multiplicity, and those at $2$ with sixteen degeneracies
and the ones at $4$ with eight degeneracies.
When $-m$ cross these points, these modes close the gap.

For the $g_P^2>0$ case, we find that the spectrum is deformed
in the imaginary direction and the eigenvalues are moved away from the real axis,
as shown in the middle panel of Fig.~\ref{fig:Dspectrum},
where the result for a typical configuration in each ensemble is plotted.
We see that a clear gap opens in the imaginary direction
at all simulated masses $m=0.5, 0,-1.5,-2$ at $g_P^2=0.2$.
This is a striking feature of SMG found only in the path-integral formalism
where the Dirac operator is a complex operator,
in contrast to the standard mass term shifting the spectrum in the real direction.
It is obvious that $k=0,8,-8$ phases are all smoothly
connected without closing the gap, which confirms the reduction $\mathbb{Z} \to \mathbb{Z}_8$ 
of the topological phases.

When $g_P^2<0$, the Yukawa interaction term is Hermitian,
rather than anti-Hermitian. 
The Dirac spectrum is deformed in the real direction,
keeping some of the eigenmodes on the real axis, as presented in the right panel of Fig.~\ref{fig:Dspectrum}.
In this case, SMG is not available, which is consistent with \cite{Fidkowski:2009dba}.
We find that the real eigenvalues are scattered along the real axis
to make the $k=0,+8,-8$ regions unclear.
This may be a signal of the first-order phase transition suggested in \cite{Fidkowski:2009dba}.

\begin{figure*}[tbhp]
\begin{center}
  \includegraphics[width=1.0\textwidth]{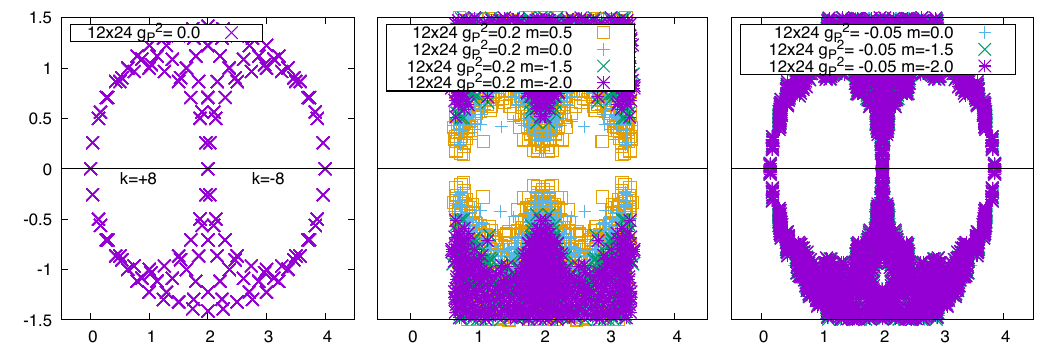}
  \caption{Eigenvalue spectrum of $D_Y(m=0)$ at $g_P^2=0.0$ (left panel), 0.2 (middle) and $-0.05$ (right).
    A clear gap in the imaginary direction is observed when $g_P^2>0$.
    For nonzero $g_P^2$, data for a typical configuration generated at different simulated values of $m$ are plotted.
  Changing mass $m$ corresponds to shifting the origin to $(-m,0)$.
  }
  \label{fig:Dspectrum}
\end{center}
\end{figure*}

\subsection{Reflection SMG with domain-walls}

Next let us examine the reflection SMG with two domain-walls located at $x=3.5$ (physical wall) and 9.5 (mirror wall)
where we make the scalar field active only near the mirror wall in the region $8\le x\le 11$.

In the free fermion case with $g_P^2=0$, the edge-localized
modes on the domain-walls
are eigenstates of $\gamma^1$ in the large volume limit
between the domain-walls.
For these modes, the mass term $m\chi^T C \chi$ disappears
and therefore, the spectrum of $D_Y^\dagger D_Y$ is gapless, due to the reflection symmetry.
Since we have 8 flavors, we have 16 such zero modes on two domain-walls.
Note that the mechanism is exactly the same as the standard domain-wall fermion
in five-dimensions where the edge modes are the eigenstates of $\gamma^5$. 

As the top panel of Fig.~\ref{fig:DdagDspectrum} shows, where the lowest eigenvector's amplitude $v^\dagger(x)v(x)$
on a typical configuration at $g_P^2=0.15$ on a $12\times 12$ lattice is plotted,
we confirm that 
most of the low-lying modes are localized either at the physical domain-wall or mirror domain-wall\footnote{There are some exceptional configurations
where a few modes have amplitudes near the both walls. We consider them as a tunneling effect when the mirror and physical modes
are accidentally degenerate.}.
Therefore, we can classify them into the physical modes and mirror modes by bias of the amplitude distribution.
In the bottom panel of Fig.~\ref{fig:DdagDspectrum}, we plot the lowest (circle symbols) and eighth (squares) mirror mode
eigenvalues of $D_Y^\dagger D_Y$ 
and those on the physical wall (cross symbols) as functions of $g_P^2$.
The gradation represents the expectation value of $\gamma^1$ of each eigenmode.
We observe that the nonzero eigenvalues have almost four-fold multiplicities (see Appendix~\ref{app:4-fold} for details).

\begin{figure*}[tbhp]
  \begin{center}
        \includegraphics[width=0.8\textwidth]{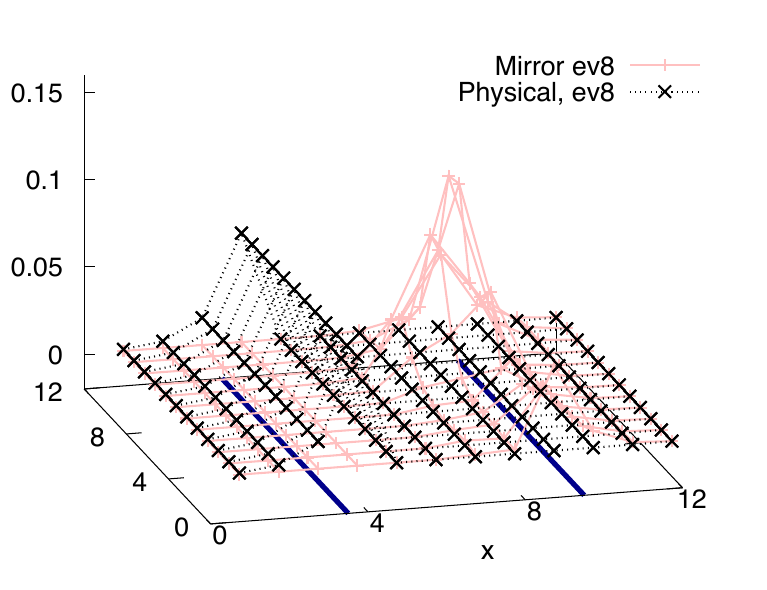}
    \includegraphics[width=0.8\textwidth]{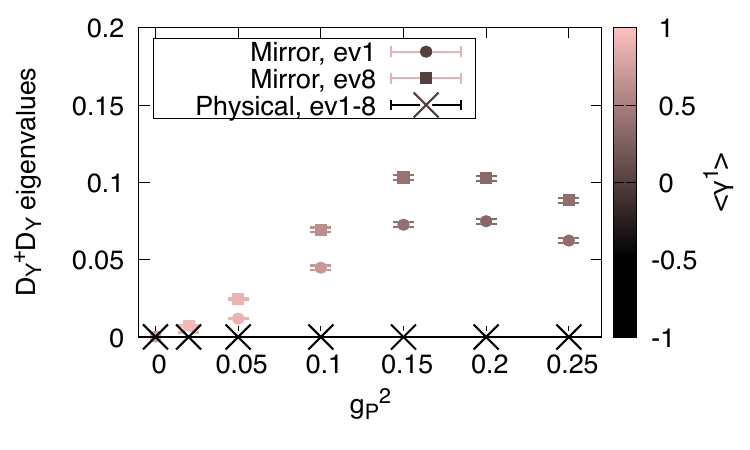}
    \caption{Top: The eighth eigenvector amplitude with $\langle \gamma^1 \rangle< 0 $ (dashed cross symbols)
    and that with $\langle \gamma^1 \rangle>0 $ (solid plus symbols) of a typical configuration at $g_P^2=0.15$ on the $12\times 12$ lattice.
    Thick lines at $x=3.5, 9.5$ show the locations of the domain-walls.
    Bottom: the spectrum of the lowest and eighth eigenvalues of $D_Y^\dagger D_Y$ on the $12\times 24$ lattice
    in the mirror sector
     (circle and square symbols, respectively)
    and those on the physical wall (crosses) as functions of $g_P^2$.
    The gradation represents the expectation value $\langle \gamma^1 \rangle$ of each eigenmode.
    }
    \label{fig:DdagDspectrum}
    \end{center}
\end{figure*}
\begin{figure*}[tbhp]
    \begin{center}
  \includegraphics[width=0.8\textwidth]{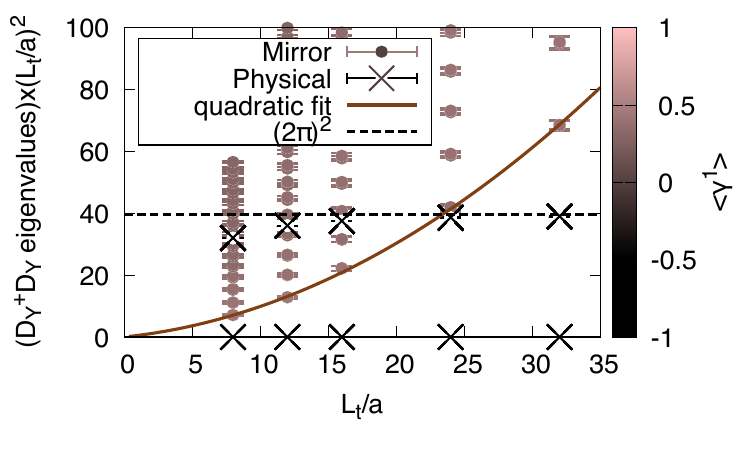}
  \caption{The eigenvalue spectrum of $D_Y^\dagger D_Y$ multiplied by $(L_t/a)^2$ as a function of $L_t/a$ at $g_P^2=0.15$.
    The physical fermion spectrum (cross symbols) is consistent with the continuum theory $(2\pi n)^2$,
    as the dashed line shows for $n=1$, which is stable against change of $L_t/a$.
    The mirror fermion eigenvalues (circles) open a gap, which scales as a function of $1/a^2$
    as represented by the solid fit curve.
  }
  \label{fig:a-scaling}
\end{center}
\end{figure*}

The gap of the mirror modes denoted by the circle and square symbols
increases as $g_P^2$, while cross symbols with eight
degeneracy stay near zero.
The physical mode's amplitude is localized near the location of the physical wall $x=3.5$
and uniformly distributed along the $t$ direction, while
the mirror modes show local lump-like structures on the mirror wall at $x=9.5$.
These data indicate that SMG by the four-Fermi interaction
selectively gaps the edge fermions on the mirror wall
while leaving the physical eigenvalue spectrum intact.

Moreover, we find that the generated gap scales as the function of $(L_t/a)^2$.
In Fig.~\ref{fig:a-scaling}, we plot the eigenvalue spectrum of
$D_Y^\dagger D_Y$ at $g_P^2=0.15$, multiplied by $(L_t/a)^2$, as a function of $L_t/a$.
The lowest eigenvalue of the mirror fermions, which are plotted by circle symbols,
scales as a function of $(L_t/a)^2$ with an $O(a)$ correction,
shown by the solid fit curve $[c L_t/a]^2(1+d a/L_t)$ with the fit parameters $c=0.230(2)$ and $d=8.7(3)$.
The physical mode spectrum (cross symbols), on the other hand, have little change
against change of $L_t/a$, which is consistent with the continuum prediction $(2\pi n)^2$.
If this $1/a^2$ scaling persists until the continuum limit,
we will be able to decouple the mirror fermions from the theory.



\section{Summary}

We have simulated eight-flavor massive Majorana fermions on a two-dimensional square lattice
with and without domain-walls, which corresponds to the path integral realization
of the Fidkowski-Kitaev's Majorana chain model.
In our simulation, the four-Fermi interaction was converted to the Yukawa interaction
with the auxiliary scalar field.
We have shown that the standard hybrid Monte Carlo method works
with reweighting the complex phase of the fermion Pfaffian,
whose fluctuation turned out to be remarkably small.

First, using the lattice Dirac operator eigenvalues without domain-walls,
we have found a striking signature of SMG in the path integral formalism,
where the gap opens in the imaginary direction of the complex eigenvalue spectrum,
in contrast to the standard mass shift in the real direction.

Then we have demonstrated that SMG works for the domain-wall fermions,
selectively gapping out the
fermions localized on the mirror wall only,
while keeping the physical wall fermion spectrum intact.
We have also confirmed that the mirror fermion spectral gap scales as $(L_t/a)^2$, while
keeping the physical eigenvalues almost constant.
These numerical spectrum of $D_Y^\dagger D_Y$ strongly
supports factorization of the massless free Majorana fermion Pfaffian
localized on the physical wall
and that of the massive fermions on the mirror wall.

In order to fully verify decoupling of the physical and mirror walls,
we need to study finite size systematics
between the two domain-walls, including measurement of the residual mass.
For this study, we give an infra-red cutoff to the physical wall
fermions,
by imposing the anti-periodic boundary condition in the temporal direction,
to make the inverse of $D_Y^\dagger D_Y$ well under control,
and investigate the finite size scaling of the fermion and boson correlators.
We find that the mirror fermions and scalar fields at $g_P^2=0.15$
have non-zero masses both in the $x$ and $t$ directions,
which are insensitive to the volume,
while the physical wall fermions
keep long-range correlation in the temporal direction.
We will report the details in a separate paper \cite{Araki:2026prep}.
 
\acknowledgments
The authors thank Lukasz Fidkowski, Yoshio Kikukawa, Ken Shiozaki, Cenke Xu and
Tatsuya Yamaoka for fruitful discussions.
We also thank Yuki Nagai and Akio Tomiya for their instruction on JuliaQCD code.
We thank the  Galileo Galilei Institute for Theoretical Physics,
where discussions at the workshop 
"Defects and Extended Excitations in Quantum Field Theory, Quantum
Matter and Statistical Models" were helpful.
We also thank the Yukawa Institute for Theoretical Physics at Kyoto University,
where this work was almost finalized during
the workshops, ``One-day workshop on symmetric mass generation'' (YITP-X-26-04),
and ``Frontiers of Lattice Fermions'' (YITP-W-26-06).
This work was
supported in part by JSPS KAKENHI Grants No.
JP23K22490, JP25K07283, and 26K07086, Japan.
This work was also supported in part by JST SPRING Grants No.
JPMJPS2138, Japan.

\appendix

\section{The original Fidkowski-Kitaev model}
\label{app:originalFK}

Here we recover the original Fidkowski-Kitaev model \cite{Fidkowski:2009dba}
from our Lagrangian given in Eq.~(\ref{eq:action}) (see also \cite{Tong:2019bbk}).
The Hamiltonian 
is obtained by the standard Legendre transformation with respect to $\chi^i = (\eta^i, \bar{\eta}^i)^T$
and its conjugate momenta.
Then we can rewrite the fields using the triality of the two 8-dimensional spinors
and a 8-dimensional vector representations of $SO(8)$ group, through the bosonization technique,
to obtain the manifestly $SO(7)$ invariant interaction term of the Hamiltonian
with the new spinor fields $\psi_j$ and $\bar{\psi_j}$:
\begin{equation}
H_{\rm int} = A \left(\sum_{j=1}^7 \psi_j\bar{\psi_j} \right)^2+B\left(\sum_{j=1}^7 \psi_j\bar{\psi_j} \right)\psi_8\bar{\psi_8},
\end{equation}  
where the coefficients are related by
\begin{equation}
g_P^2=\frac{1}{4}A-\frac{1}{8}B,\;\;\; g_S^2=\frac{5}{4}A+\frac{3}{8}B.
\end{equation}

According to \cite{Fidkowski:2009dba}, SMG works in the parameter region $A>B/2$ with $B<0$.
In this work, we fix $g_S^2=0$, which is irrelevant to SMG,
and consider the $A=-3B/10$ line.
We expect that SMG works for $g_P^2=-B/5>0$ and does not for $g_P^2\le 0$.


\section{Gap of $D_Y^\dagger D_Y$ and mass}\label{app:gap}

Suppose that $D_Y^\dagger D_Y$ has a gap $M_0^2>0$, which is sufficiently
larger than the inverse volume size $1/L$,
for any field configurations of the scalar $\pi^\alpha$ generated by a HMC simulation. 
Here we show that the gap guarantees that
fermion correlator exponentially decays
at sufficiently long distances.

Let us denote the fermion two-point function by
\begin{equation}
S_F(x-z):=\langle x|(CD_Y)^{-1}|z\rangle =  \langle x|(D_Y^\dagger D_Y)^{-1}\sum_y |y\rangle \langle y|D_Y^\dagger C^{-1}|z\rangle,
\end{equation}
where we have suppressed irrelevant spinor or flavor indices.
Since $D_Y^\dagger C^{-1}$ is a local differential operator, $\langle y|D_Y^\dagger C^{-1}|z\rangle$
has a finite support only nearby $y=z$.
Therefore, it is enough to consider $\langle x|(D_Y^\dagger D_Y)^{-1}|y\rangle$ for large $|x-y|$.

According to \cite{HernandezJansenLuscher},
exponential decay of inverse holds for
a general positive bounded ultra-local operator on a $d$-dimensional square lattice,
using the taxi-driver distance
\begin{equation}
  d_1(x,y):=\lVert x-y\rVert_1
  =\sum_{j=1}^d|x_j-y_j|.
\end{equation}
Suppose $H$ to be such a lattice operator, whose eigenvalues are
in a finite range $[u,v]$ where $0<u<v<\infty$.
We further assume that there exists a positive number $R$
such that the kernel $H(x,y)=\langle x|H|y\rangle=0$
whenever $d_1(x,y)>R$.
Define the kernel of the inverse by
\begin{equation}
  G(x,y):=\langle x|H^{-1} |y\rangle.
\end{equation}
The norm $\lVert G(x,y)\rVert$ below is the operator norm on the internal
space.
In our case, $H=D_Y^\dagger D_Y$, $R=4$, $u=M_0^2$ and $v$ is the maximal eigenvalue of $D_Y^\dagger D_Y$.

\begin{proposition}[Exponential locality of the inverse]\label{prop:local}
  With the above assumptions, $G(x,y)$ satisfies a bound
  \begin{equation}
    \lVert G(x,y)\rVert
    \leq \frac{1}{u}\exp\!\bigl[-\mu_0 d_1(x,y)\bigr]
    \qquad (x,y\in\Lambda),
    \label{eq:main-bound}
  \end{equation}
  where
  \begin{equation}
    \mu_0
    =\frac{1}{R}\log\!\left(\frac{v+u}{v-u}\right)>0.
    \label{eq:mu-geometric}
  \end{equation}
  In particular, both the prefactor and the localization length are uniform
  whenever $R,u,v$ are uniform.
\end{proposition}

The proof rests on the following abstract version of the finite-propagation
argument used in Section~2.2 of Ref.~\cite{HernandezJansenLuscher}.

\begin{lemma}\label{lem:polynomial-locality}
  Let $Z$ be a self-adjoint and ultralocal operator 
  and suppose that
  \[
    \spec(Z)\subset[-1,1].
  \]
  For each $k\geq0$, let $F_k$ be a polynomial of degree at most $k$ such that
  \[
    \sup_{|z|\leq1}|F_k(z)|\leq1.
  \]
  Let $K>0$ and $0<t<1$, and define the norm-convergent series
  \begin{equation}
    \mathcal G:=K\sum_{k=0}^{\infty}t^kF_k(Z).
    \label{eq:G-series}
  \end{equation}
  Then
  \begin{equation}
    \lVert \mathcal G(x,y)\rVert
    \leq \frac{K}{1-t}
      \exp\!\bigl[-\mu d_1(x,y)\bigr],
    \qquad
    \mu=-\frac{1}{R}\log t>0.
    \label{eq:lemma-bound}
  \end{equation}
\end{lemma}

\begin{proof}
  Set
  \[
    \mathcal G_k(x,y):=\langle x | F_k(Z)|y\rangle.
  \]
  Since $Z$ is self-adjoint and $\spec(Z)\subset[-1,1]$, the continuous
  functional calculus gives
  \[
    \lVert F_k(Z)\rVert
    \leq \sup_{|z|\leq1}|F_k(z)|\leq1,
  \]
  and hence
  \begin{equation}
    \lVert \mathcal G_k(x,y)\rVert\leq1.
    \label{eq:Gk-norm}
  \end{equation}

  A product of $j$ operators of range $R$ has range at most $jR$.
  Consequently, because $F_k$ has degree at most $k$,
  \begin{equation}
    \mathcal G_k(x,y)=0
    \qquad\text{if}\qquad
    d_1(x,y)>kR.
    \label{eq:Gk-range}
  \end{equation}
  Introduce the ceil function $\lceil \cdot \rceil$ and define
  \begin{equation}
    N:=\left\lceil\frac{d_1(x,y)}{R}\right\rceil.
    \label{eq:N-ceiling}
  \end{equation}
  Equation~\eqref{eq:Gk-range} implies that all terms with $k<N$ vanish.
  Therefore, using \eqref{eq:Gk-norm},
  \begin{align}
    \lVert \mathcal G(x,y)\rVert
    &\leq K\sum_{k=N}^{\infty}t^k
      =\frac{K}{1-t}t^N \nonumber\\
    &\leq \frac{K}{1-t}t^{d_1(x,y)/R}
      =\frac{K}{1-t}
       \exp\!\left[\frac{\log t}{R}d_1(x,y)\right].
  \end{align}
  Here the second inequality uses $N\geq d_1(x,y)/R$ together with
  $0<t<1$. This is \eqref{eq:lemma-bound}.
\end{proof}

Using this lemma, we can prove Proposition~\ref{prop:local}.
Set
\begin{equation}
  c:=\frac{u+v}{2},
  \qquad
  t:=\frac{v-u}{v+u},
  \qquad
  Z:=\frac{1}{t}\left(I-\frac{H}{c}\right)
    =\frac{v+u-2H}{v-u}.
  \label{eq:ctZ}
\end{equation}
Then $c>0$, $0<t<1$, and the spectral bounds on $H$ imply
\[
  \spec(Z)\subset[-1,1].
\]
Since $H=c(I-tZ)$, the Neumann series converges in operator
norm and gives
\begin{equation}
  H^{-1}
  =\frac{1}{c}(I-tZ)^{-1}
  =\frac{1}{c}\sum_{k=0}^{\infty}t^kZ^k.
  \label{eq:neumann-expansion}
\end{equation}
Apply Lemma~\ref{lem:polynomial-locality} with
\[
  K=\frac{1}{c},
  \qquad
  F_k(z)=z^k.
\]
This yields
\[
  \lVert H^{-1}(x,y)\rVert
  \leq \frac{1}{c(1-t)}
  \exp\!\left[-\frac{-\log t}{R}d_1(x,y)\right].
\]
Finally,
\[
  \frac{1}{c(1-t)}=\frac{1}{u},
  \qquad
  -\frac{1}{R}\log t
  =\frac{1}{R}\log\!\left(\frac{v+u}{v-u}\right),
\]
which proves \eqref{eq:main-bound}.

In this work, we examine SMG in two different setups
with and without domain-walls in the mass term.
Without domain-walls, the gap in $D_Y^\dagger D_Y$ guarantees
by the above discussion
that all the Majorana fermions are massive.

It is, however, nontrivial in the domain-wall fermion cases,
where we observe that the spectrum of $D_Y^\dagger D_Y$
is decomposed into that of the massless free fermions on the physical wall
and the one of the gapped interacting fermions on the mirror wall.
In order to verify decoupling of the two fermions,
we need to give an infra-red cut-off to the physical fermions,
by imposing the anti-periodic boundary condition, for example,
and confirm that $S_F(x-z)$ between two walls exponentially decay
and check the finite size scaling at different lattice volumes.
It is also important to confirm that the correlators of
the scalar field $\pi^\alpha$ exponentially decay.

\section{Symmetry of $D_Y$ and degeneracy of the spectrum}\label{app:4-fold}

When $g_P^2=0$, every eigenvalue of the Dirac operator $D_Y$ has, at least,
eight multiplicities due to the $SO(8)$ flavor symmetry.

For nonzero Yukawa interaction with $g_P^2 \neq 0$,
$D_Y$ has the charge-conjugation symmetry:
\begin{equation}
  CD_YC^{-1} = D_Y^T,
\end{equation}
where the superscript $T$ denotes the transpose.
For any right eigenfunction $\phi$ with $D_Y \phi=\lambda \phi$ where $\lambda$ is the complex eigenvalue,
$\phi^T C$ is a left eigenfunction satisfying $\phi^T C D = \lambda \phi^T C $.
Since $\phi^T C \phi=0$, there must be another right eigenfunction $\phi'$ with the same eigenvalue $\lambda$
which has nonzero $\phi^T C \phi'$.
Therefore, any eigenvalue $\lambda$ of $D_Y$ has, at least, two-fold degeneracy.

When $g_P^2\le 0$, the $D_Y$ has the $\gamma_3$ Hermiticity:
\begin{equation}
  \gamma_3 D_Y\gamma_3 = D_Y^\dagger.
\end{equation}
In this case, for any right eigenfunction $\phi$ with $D_Y \phi=\lambda \phi$ where $\lambda$ is the complex eigenvalue,
$\phi^\dagger \gamma_5$ is a left eigenfunction satisfying $\phi^\dagger \gamma_5 D_Y = \lambda^* \phi^\dagger \gamma_5$.
Therefore, every complex eigenvalue makes a pair with its conjugate,
which makes the Pfaffian of $D_Y$ always real.

For $g_P^2>0$ the $\gamma_3$ Hermiticity is broken but we still numerically observe
a good pairing of $\lambda$ and $\lambda^*$ in our simulated ensembles.
The difference is within 1\% but it is finite, which adds up to
small but nonzero value of the phase of the fermion Pfaffian.
It would be interesting to study if this small phase is
related to absence of the 'tHooft anomaly.

In the spectrum of $D_Y^\dagger D_Y$, we observe (almost) four multiplicities
in the mirror wall fermion sector, while it is eight-fold degeneracies in the physical wall fermions.

\section{Summary of simulation parameters and numerical data}

Here we summarize the simulation parameters and numerical results.
In particular, we show that the complex Pfaffian phase is kept small in all our simulated ensembles.

In Tab.~\ref{tab:NoDW}, our simulations with uniform masses without domain-walls are summarized.
We also present the results for the Pfaffian phase $\theta={\rm arg}({\rm Pf}[CD_Y])$, its standard deviation
$\sqrt{\langle \Delta\theta^2 \rangle}$ and the lowest eigenvalue of $D_Y^\dagger D_Y$.
For $g_P^2=-0.05$, we do not present the results, since $\theta$ is trivially zero due to the $\gamma^3$ Hermiticity and
the gap of $(D_Y^\dagger D_Y)$ is closed at some value of $m$, due to the presence of real eigenvalues of $D_Y$.

Table~\ref{tab:DW} presents the simulation set up with domain-walls on our main lattice $L_x/a \times L_y/a= 12\times 24$.
We present the lowest eigenvalue of $(D_Y^\dagger D_Y)$ in the mirror sector $\langle \lambda_0^{\rm mirr}\rangle$ as well as
  the expectation value  $\langle \gamma^1 \rangle_{\lambda_0^{\rm mirr}} $ of the lowest mode,
  and those $\langle \lambda_0^{\rm phys}\rangle$ and $\langle \gamma^1 \rangle_{\lambda_0^{\rm phys}} $ in the physical sector.
  We observe that the mirror gap $\langle \lambda_0^{\rm mirr}\rangle$ becomes nonzero,
  while the physical fermion is kept massless.
  Interestingly, the lowest eigenvalue $\sim 2.0\times 10^{-8}$ at $g_P^2=0.15$
is closer to zero than that $\sim 1.4\times 10^{-4}$ at $g_P^2=0.0$.
  It is also interesting to see that $\langle \lambda_0^{\rm mirr}\rangle$ does not monotonically increases with $g_P^2$
  but reaches a maximum around $g_P^2=0.15$--0.2.

In Tab.~\ref{tab:DW-Vol} we present $L_t/a$ dependence of the numerical results at a fixed value of $g_P^2=0.15$.
The first excited eigenvalue  in the physical sector $\langle \lambda_1^{\rm phys}\rangle  a^2$  scales
as $1/L_t^2$, while the lowest eigenvalue in the mirror sector is consistent with a constant plus $O(a/L_t)$ corrections.

\begin{table}[tbh]
\centering
\begin{tabular}{c c c c c | c c | c}
  \hline\hline
  $g_P^2$ & $L_x/a$ & $L_t/a$ & $m$ & $N_{\rm trj}$
  & $\langle \theta \rangle$[rad] & $\sqrt{\langle \Delta\theta^2 \rangle}$[rad] & $\langle \lambda_0\rangle a^2$\\
  \hline
  0.2 & 12 & 24 & 0.5 & 1000 & $-5(7)\times 10^{-7}$ & $7.1(5)\times 10^{-6}$ & 0.770(2)\\
  &    &    & 0.0 & 1000 &  $0.5(1.9)\times 10^{-5}$ & $1.4(2)\times 10^{-4}$ & 0.256(4)\\
  &    &    & -1.5 & 1000 & $2(5)\times 10^{-4}$ & $4.4(2)\times 10^{-3}$ & 0.083(3)\\
  &    &    & -2.0 & 1000 & $0$ & $0$ & 0.099(3)\\
  \hline
  -0.05 & 12 & 24 & 0.0 & 1000 & -- & -- & -- \\
  &    &    & -1.5 & 1000 & -- & -- & --\\
  &    &    & -2.0 & 1000 & -- & -- & --\\
  \hline
\end{tabular}  
\caption{Simulation parameters without domain-walls and results for the Pfaffian phase $\langle \theta \rangle$,
  $\sqrt{\langle \Delta\theta^2 \rangle}$, and the minimal value of the eigenvalue $\langle \lambda_0\rangle$ of $(D_Y^\dagger D_Y)$.
  For $g_P^2=-0.05$, we do not present the results, since $\theta$ is trivially zero due to the $\gamma^3$ Hermiticity and
  the gap of $(D_Y^\dagger D_Y)$ is closed due to the presence of real eigenvalues of $D_Y$.
}\label{tab:NoDW}
\end{table}

\begin{table}[tbh]
\centering
\begin{tabular}{c c | c c | c c c c}
  \hline\hline
  $g_P^2$ & $N_{\rm trj}$
  & $\langle \theta \rangle$[rad] & $\sqrt{\langle \Delta\theta^2 \rangle}$[rad]
  & $\langle \lambda_0^{\rm mirr}\rangle  a^2$ & $\langle \gamma^1 \rangle_{\lambda_0^{\rm mirr}} $
  & $\langle \lambda_0^{\rm phys}\rangle  a^2$ & $\langle \gamma^1 \rangle_{\lambda_0^{\rm phys}} $\\
  \hline
  0.02 & 1000 & $-2.6(1.9)\times 10^{-5}$ & $2.2(2)\times 10^{-4}$ & 0.0031(4) & 0.8798(4) & 1.64(9)$\times 10^{-6}$ & -0.96852208(9)\\
  0.05 & 3000 & $-0.00088(59)$ & $0.0074(7)$ & 0.0119(3) & 0.8875(3) & 1.56(4)$\times 10^{-7}$ & -0.993769983(4)\\
  0.10 & 3000 & 0.0013(12) & 0.015(3) & 0.0448(1) & 0.671(1) & 2.61(9)$\times 10^{-8}$ & -0.9987171237(9)\\
  0.15 & 3000 & -0.0003(1) & 0.0026(7) & 0.073(2) & 0.371(2) & 2.48(8)$\times 10^{-8}$ &  -0.9994428155(8)\\
  0.20 & 3000 & -1.6(6.5)$\times 10^{-5}$ & 0.0010(1) & 0.075(2) & 0.297(2) & 2.76(8)$\times 10^{-8}$ & -0.9996182654(8)\\
  0.25 & 3000 & -1.7(5.1)$\times 10^{-5}$ & 0.00077(6) & 0.062(2) & 0.352(2) & 2.92(9)$\times 10^{-8}$ & -0.9996659203(9)\\
   \hline
\end{tabular}  
\caption{Simulation parameters with two domain-walls on a fixed lattice $12\times 24$
  and some numerical results. Here the Pfaffian phase $\langle \theta \rangle$,
  $\sqrt{\langle \Delta\theta^2 \rangle}$, the lowest eigenvalue of $(D_Y^\dagger D_Y)$ in the mirror sector $\langle \lambda_0^{\rm mirr}\rangle$ as well as
  the expectation value  $\langle \gamma^1 \rangle_{\lambda_0^{\rm mirr}} $ of the lowest mode,
  and those $\langle \lambda_0^{\rm phys}\rangle$, $\langle \gamma^1 \rangle_{\lambda_0^{\rm phys}} $ in the physical sector
are presented.
}\label{tab:DW}
\end{table}

\begin{table}[tbh]
\centering
\begin{tabular}{c c c c | c c | c c c c}
  \hline\hline
  $g_P^2$ & $L_x/a$ & $L_t/a$ & $N_{\rm trj}$
  & $\langle \theta \rangle$[rad] & $\sqrt{\langle \Delta\theta^2 \rangle}$[rad]
  & $\langle \lambda_0^{\rm mirr}\rangle  a^2$ & $\langle \gamma^1 \rangle_{\lambda_0^{\rm mirr}} $
  & $\langle \lambda_1^{\rm phys}\rangle  a^2$ & $\langle \gamma^1 \rangle_{\lambda_1^{\rm phys}} $\\
  \hline
  0.15 & 12 & 8 & 1000 & -0.00017(13) & 0.0012(2) & 0.111(5) & 0.426(5) & 0.49945(8) & -0.88896(8)\\
  &    & 12 & 1700 & -0.0002(2) & 0.0015(3) & 0.089(2) & 0.434(2) & 0.24977(1) & -0.98010(1)\\
  &    & 16 & 1000 & -0.0001(3) & 0.0022(5) & 0.087(3) & 0.406(3) & 0.14614(4) & -0.93882(4)\\
  &    & 24 & 3000 & -0.0003(1) & 0.0026(7) & 0.073(2) & 0.371(2) & 0.06678(1) & -0.94942(1)\\
  &    & 32 & 3000 &  0.0001(2) & 0.0022(2) & 0.067(1) & 0.389(1) & 0.03795(1) & -0.98634(1)\\
   \hline
\end{tabular}  
\caption{Simulation parameters with two domain-walls on different volume lattices 
  and some numerical results at a fixed $g_P^2=0.15$.
  The Pfaffian phase $\langle \theta \rangle$,
  $\sqrt{\langle \Delta\theta^2 \rangle}$, the lowest eigenvalue of $(D_Y^\dagger D_Y)$ in the mirror sector $\langle \lambda_0^{\rm mirr}\rangle$ as well as
  the expectation value  $\langle \gamma^1 \rangle_{\lambda_0^{\rm mirr}} $ of the lowest mode,
  and those for the first excited state, $\langle \lambda_1^{\rm phys}\rangle$ and $\langle \gamma^1 \rangle_{\lambda_1^{\rm phys}} $ in the physical sector
  are presented.
}\label{tab:DW-Vol}
\end{table}

\clearpage
\bibliography{ref.bib}

\begin{thebibliography}{49}%
\makeatletter
\providecommand \@ifxundefined [1]{%
 \@ifx{#1\undefined}
}%
\providecommand \@ifnum [1]{%
 \ifnum #1\expandafter \@firstoftwo
 \else \expandafter \@secondoftwo
 \fi
}%
\providecommand \@ifx [1]{%
 \ifx #1\expandafter \@firstoftwo
 \else \expandafter \@secondoftwo
 \fi
}%
\providecommand \natexlab [1]{#1}%
\providecommand \enquote  [1]{``#1''}%
\providecommand \bibnamefont  [1]{#1}%
\providecommand \bibfnamefont [1]{#1}%
\providecommand \citenamefont [1]{#1}%
\providecommand \href@noop [0]{\@secondoftwo}%
\providecommand \href [0]{\begingroup \@sanitize@url \@href}%
\providecommand \@href[1]{\@@startlink{#1}\@@href}%
\providecommand \@@href[1]{\endgroup#1\@@endlink}%
\providecommand \@sanitize@url [0]{\catcode `\\12\catcode `\$12\catcode `\&12\catcode `\#12\catcode `\^12\catcode `\_12\catcode `\%12\relax}%
\providecommand \@@startlink[1]{}%
\providecommand \@@endlink[0]{}%
\providecommand \url  [0]{\begingroup\@sanitize@url \@url }%
\providecommand \@url [1]{\endgroup\@href {#1}{\urlprefix }}%
\providecommand \urlprefix  [0]{URL }%
\providecommand \Eprint [0]{\href }%
\providecommand \doibase [0]{https://doi.org/}%
\providecommand \selectlanguage [0]{\@gobble}%
\providecommand \bibinfo  [0]{\@secondoftwo}%
\providecommand \bibfield  [0]{\@secondoftwo}%
\providecommand \translation [1]{[#1]}%
\providecommand \BibitemOpen [0]{}%
\providecommand \bibitemStop [0]{}%
\providecommand \bibitemNoStop [0]{.\EOS\space}%
\providecommand \EOS [0]{\spacefactor3000\relax}%
\providecommand \BibitemShut  [1]{\csname bibitem#1\endcsname}%
\let\auto@bib@innerbib\@empty
\bibitem [{\citenamefont {Kaplan}(1992)}]{Kaplan:1992bt}%
  \BibitemOpen
  \bibfield  {author} {\bibinfo {author} {\bibfnamefont {D.~B.}\ \bibnamefont {Kaplan}},\ }\bibfield  {title} {\bibinfo {title} {{A Method for simulating chiral fermions on the lattice}},\ }\href {https://doi.org/10.1016/0370-2693(92)91112-M} {\bibfield  {journal} {\bibinfo  {journal} {Phys. Lett. B}\ }\textbf {\bibinfo {volume} {288}},\ \bibinfo {pages} {342} (\bibinfo {year} {1992})},\ \Eprint {https://arxiv.org/abs/hep-lat/9206013} {arXiv:hep-lat/9206013} \BibitemShut {NoStop}%
\bibitem [{\citenamefont {Shamir}(1993)}]{Shamir:1993zy}%
  \BibitemOpen
  \bibfield  {author} {\bibinfo {author} {\bibfnamefont {Y.}~\bibnamefont {Shamir}},\ }\bibfield  {title} {\bibinfo {title} {{Chiral fermions from lattice boundaries}},\ }\href {https://doi.org/10.1016/0550-3213(93)90162-I} {\bibfield  {journal} {\bibinfo  {journal} {Nucl. Phys. B}\ }\textbf {\bibinfo {volume} {406}},\ \bibinfo {pages} {90} (\bibinfo {year} {1993})},\ \Eprint {https://arxiv.org/abs/hep-lat/9303005} {arXiv:hep-lat/9303005} \BibitemShut {NoStop}%
\bibitem [{\citenamefont {Neuberger}(1998)}]{Neuberger:1997fp}%
  \BibitemOpen
  \bibfield  {author} {\bibinfo {author} {\bibfnamefont {H.}~\bibnamefont {Neuberger}},\ }\bibfield  {title} {\bibinfo {title} {{Exactly massless quarks on the lattice}},\ }\href {https://doi.org/10.1016/S0370-2693(97)01368-3} {\bibfield  {journal} {\bibinfo  {journal} {Phys. Lett. B}\ }\textbf {\bibinfo {volume} {417}},\ \bibinfo {pages} {141} (\bibinfo {year} {1998})},\ \Eprint {https://arxiv.org/abs/hep-lat/9707022} {arXiv:hep-lat/9707022} \BibitemShut {NoStop}%
\bibitem [{\citenamefont {Ginsparg}\ and\ \citenamefont {Wilson}(1982)}]{Ginsparg:1981bj}%
  \BibitemOpen
  \bibfield  {author} {\bibinfo {author} {\bibfnamefont {P.~H.}\ \bibnamefont {Ginsparg}}\ and\ \bibinfo {author} {\bibfnamefont {K.~G.}\ \bibnamefont {Wilson}},\ }\bibfield  {title} {\bibinfo {title} {{A Remnant of Chiral Symmetry on the Lattice}},\ }\href {https://doi.org/10.1103/PhysRevD.25.2649} {\bibfield  {journal} {\bibinfo  {journal} {Phys. Rev. D}\ }\textbf {\bibinfo {volume} {25}},\ \bibinfo {pages} {2649} (\bibinfo {year} {1982})}\BibitemShut {NoStop}%
\bibitem [{\citenamefont {Luscher}(1998)}]{Luscher:1998pqa}%
  \BibitemOpen
  \bibfield  {author} {\bibinfo {author} {\bibfnamefont {M.}~\bibnamefont {Luscher}},\ }\bibfield  {title} {\bibinfo {title} {{Exact chiral symmetry on the lattice and the Ginsparg-Wilson relation}},\ }\href {https://doi.org/10.1016/S0370-2693(98)00423-7} {\bibfield  {journal} {\bibinfo  {journal} {Phys. Lett. B}\ }\textbf {\bibinfo {volume} {428}},\ \bibinfo {pages} {342} (\bibinfo {year} {1998})},\ \Eprint {https://arxiv.org/abs/hep-lat/9802011} {arXiv:hep-lat/9802011} \BibitemShut {NoStop}%
\bibitem [{\citenamefont {Luscher}(1999)}]{Luscher:1998du}%
  \BibitemOpen
  \bibfield  {author} {\bibinfo {author} {\bibfnamefont {M.}~\bibnamefont {Luscher}},\ }\bibfield  {title} {\bibinfo {title} {{Abelian chiral gauge theories on the lattice with exact gauge invariance}},\ }\href {https://doi.org/10.1016/S0550-3213(99)00115-7} {\bibfield  {journal} {\bibinfo  {journal} {Nucl. Phys. B}\ }\textbf {\bibinfo {volume} {549}},\ \bibinfo {pages} {295} (\bibinfo {year} {1999})},\ \Eprint {https://arxiv.org/abs/hep-lat/9811032} {arXiv:hep-lat/9811032} \BibitemShut {NoStop}%
\bibitem [{\citenamefont {Kikukawa}\ and\ \citenamefont {Nakayama}(2001)}]{Kikukawa:2000kd}%
  \BibitemOpen
  \bibfield  {author} {\bibinfo {author} {\bibfnamefont {Y.}~\bibnamefont {Kikukawa}}\ and\ \bibinfo {author} {\bibfnamefont {Y.}~\bibnamefont {Nakayama}},\ }\bibfield  {title} {\bibinfo {title} {{Gauge anomaly cancellations in SU(2)(L) x U(1)(Y) electroweak theory on the lattice}},\ }\href {https://doi.org/10.1016/S0550-3213(00)00714-8} {\bibfield  {journal} {\bibinfo  {journal} {Nucl. Phys. B}\ }\textbf {\bibinfo {volume} {597}},\ \bibinfo {pages} {519} (\bibinfo {year} {2001})},\ \Eprint {https://arxiv.org/abs/hep-lat/0005015} {arXiv:hep-lat/0005015} \BibitemShut {NoStop}%
\bibitem [{\citenamefont {Aoki}\ and\ \citenamefont {Fukaya}(2022)}]{Aoki:2022cwg}%
  \BibitemOpen
  \bibfield  {author} {\bibinfo {author} {\bibfnamefont {S.}~\bibnamefont {Aoki}}\ and\ \bibinfo {author} {\bibfnamefont {H.}~\bibnamefont {Fukaya}},\ }\bibfield  {title} {\bibinfo {title} {{Curved domain-wall fermions}},\ }\href {https://doi.org/10.1093/ptep/ptac075} {\bibfield  {journal} {\bibinfo  {journal} {PTEP}\ }\textbf {\bibinfo {volume} {2022}},\ \bibinfo {pages} {063B04} (\bibinfo {year} {2022})},\ \Eprint {https://arxiv.org/abs/2203.03782} {arXiv:2203.03782 [hep-lat]} \BibitemShut {NoStop}%
\bibitem [{\citenamefont {Aoki}\ and\ \citenamefont {Fukaya}(2023)}]{Aoki:2022aez}%
  \BibitemOpen
  \bibfield  {author} {\bibinfo {author} {\bibfnamefont {S.}~\bibnamefont {Aoki}}\ and\ \bibinfo {author} {\bibfnamefont {H.}~\bibnamefont {Fukaya}},\ }\bibfield  {title} {\bibinfo {title} {{Curved domain-wall fermion and its anomaly inflow}},\ }\href {https://doi.org/10.1093/ptep/ptad023} {\bibfield  {journal} {\bibinfo  {journal} {PTEP}\ }\textbf {\bibinfo {volume} {2023}},\ \bibinfo {pages} {033B05} (\bibinfo {year} {2023})},\ \Eprint {https://arxiv.org/abs/2212.11583} {arXiv:2212.11583 [hep-lat]} \BibitemShut {NoStop}%
\bibitem [{\citenamefont {Kaplan}(2024)}]{Kaplan:2023pxd}%
  \BibitemOpen
  \bibfield  {author} {\bibinfo {author} {\bibfnamefont {D.~B.}\ \bibnamefont {Kaplan}},\ }\bibfield  {title} {\bibinfo {title} {{Chiral Gauge Theory at the Boundary between Topological Phases}},\ }\href {https://doi.org/10.1103/PhysRevLett.132.141603} {\bibfield  {journal} {\bibinfo  {journal} {Phys. Rev. Lett.}\ }\textbf {\bibinfo {volume} {132}},\ \bibinfo {pages} {141603} (\bibinfo {year} {2024})},\ \Eprint {https://arxiv.org/abs/2312.01494} {arXiv:2312.01494 [hep-lat]} \BibitemShut {NoStop}%
\bibitem [{\citenamefont {Kaplan}\ and\ \citenamefont {Sen}(2024)}]{Kaplan:2023pvd}%
  \BibitemOpen
  \bibfield  {author} {\bibinfo {author} {\bibfnamefont {D.~B.}\ \bibnamefont {Kaplan}}\ and\ \bibinfo {author} {\bibfnamefont {S.}~\bibnamefont {Sen}},\ }\bibfield  {title} {\bibinfo {title} {{Weyl Fermions on a Finite Lattice}},\ }\href {https://doi.org/10.1103/PhysRevLett.132.141604} {\bibfield  {journal} {\bibinfo  {journal} {Phys. Rev. Lett.}\ }\textbf {\bibinfo {volume} {132}},\ \bibinfo {pages} {141604} (\bibinfo {year} {2024})},\ \Eprint {https://arxiv.org/abs/2312.04012} {arXiv:2312.04012 [hep-lat]} \BibitemShut {NoStop}%
\bibitem [{\citenamefont {Aoki}\ \emph {et~al.}(2024)\citenamefont {Aoki}, \citenamefont {Fukaya},\ and\ \citenamefont {Kan}}]{Aoki:2024bwx}%
  \BibitemOpen
  \bibfield  {author} {\bibinfo {author} {\bibfnamefont {S.}~\bibnamefont {Aoki}}, \bibinfo {author} {\bibfnamefont {H.}~\bibnamefont {Fukaya}},\ and\ \bibinfo {author} {\bibfnamefont {N.}~\bibnamefont {Kan}},\ }\bibfield  {title} {\bibinfo {title} {{A Lattice Formulation of Weyl Fermions on a Single Curved Surface}},\ }\href {https://doi.org/10.1093/ptep/ptae041} {\bibfield  {journal} {\bibinfo  {journal} {PTEP}\ }\textbf {\bibinfo {volume} {2024}},\ \bibinfo {pages} {043B05} (\bibinfo {year} {2024})},\ \Eprint {https://arxiv.org/abs/2402.09774} {arXiv:2402.09774 [hep-lat]} \BibitemShut {NoStop}%
\bibitem [{\citenamefont {Clancy}\ and\ \citenamefont {Kaplan}(2025)}]{Clancy:2024bjb}%
  \BibitemOpen
  \bibfield  {author} {\bibinfo {author} {\bibfnamefont {M.}~\bibnamefont {Clancy}}\ and\ \bibinfo {author} {\bibfnamefont {D.~B.}\ \bibnamefont {Kaplan}},\ }\bibfield  {title} {\bibinfo {title} {{Chiral edge states on spheres for lattice domain wall fermions}},\ }\href {https://doi.org/10.1103/PhysRevD.111.L031503} {\bibfield  {journal} {\bibinfo  {journal} {Phys. Rev. D}\ }\textbf {\bibinfo {volume} {111}},\ \bibinfo {pages} {L031503} (\bibinfo {year} {2025})},\ \Eprint {https://arxiv.org/abs/2410.23065} {arXiv:2410.23065 [hep-lat]} \BibitemShut {NoStop}%
\bibitem [{\citenamefont {Eichten}\ and\ \citenamefont {Preskill}(1986)}]{Eichten:1985ft}%
  \BibitemOpen
  \bibfield  {author} {\bibinfo {author} {\bibfnamefont {E.}~\bibnamefont {Eichten}}\ and\ \bibinfo {author} {\bibfnamefont {J.}~\bibnamefont {Preskill}},\ }\bibfield  {title} {\bibinfo {title} {{Chiral Gauge Theories on the Lattice}},\ }\href {https://doi.org/10.1016/0550-3213(86)90207-5} {\bibfield  {journal} {\bibinfo  {journal} {Nucl. Phys. B}\ }\textbf {\bibinfo {volume} {268}},\ \bibinfo {pages} {179} (\bibinfo {year} {1986})}\BibitemShut {NoStop}%
\bibitem [{\citenamefont {Creutz}\ \emph {et~al.}(1997)\citenamefont {Creutz}, \citenamefont {Tytgat}, \citenamefont {Rebbi},\ and\ \citenamefont {Xue}}]{Creutz:1996xc}%
  \BibitemOpen
  \bibfield  {author} {\bibinfo {author} {\bibfnamefont {M.}~\bibnamefont {Creutz}}, \bibinfo {author} {\bibfnamefont {M.}~\bibnamefont {Tytgat}}, \bibinfo {author} {\bibfnamefont {C.}~\bibnamefont {Rebbi}},\ and\ \bibinfo {author} {\bibfnamefont {S.-S.}\ \bibnamefont {Xue}},\ }\bibfield  {title} {\bibinfo {title} {{Lattice formulation of the standard model}},\ }\href {https://doi.org/10.1016/S0370-2693(97)00463-2} {\bibfield  {journal} {\bibinfo  {journal} {Phys. Lett. B}\ }\textbf {\bibinfo {volume} {402}},\ \bibinfo {pages} {341} (\bibinfo {year} {1997})},\ \Eprint {https://arxiv.org/abs/hep-lat/9612017} {arXiv:hep-lat/9612017} \BibitemShut {NoStop}%
\bibitem [{\citenamefont {Fidkowski}\ and\ \citenamefont {Kitaev}(2010)}]{Fidkowski:2009dba}%
  \BibitemOpen
  \bibfield  {author} {\bibinfo {author} {\bibfnamefont {L.}~\bibnamefont {Fidkowski}}\ and\ \bibinfo {author} {\bibfnamefont {A.}~\bibnamefont {Kitaev}},\ }\bibfield  {title} {\bibinfo {title} {{The effects of interactions on the topological classification of free fermion systems}},\ }\href {https://doi.org/10.1103/PhysRevB.81.134509} {\bibfield  {journal} {\bibinfo  {journal} {Phys. Rev. B}\ }\textbf {\bibinfo {volume} {81}},\ \bibinfo {pages} {134509} (\bibinfo {year} {2010})},\ \Eprint {https://arxiv.org/abs/0904.2197} {arXiv:0904.2197 [cond-mat.str-el]} \BibitemShut {NoStop}%
\bibitem [{\citenamefont {Fidkowski}\ and\ \citenamefont {Kitaev}(2011)}]{Fidkowski:2010jmn}%
  \BibitemOpen
  \bibfield  {author} {\bibinfo {author} {\bibfnamefont {L.}~\bibnamefont {Fidkowski}}\ and\ \bibinfo {author} {\bibfnamefont {A.}~\bibnamefont {Kitaev}},\ }\bibfield  {title} {\bibinfo {title} {{Topological phases of fermions in one dimension}},\ }\href {https://doi.org/10.1103/PhysRevB.83.075103} {\bibfield  {journal} {\bibinfo  {journal} {Phys. Rev. B}\ }\textbf {\bibinfo {volume} {83}},\ \bibinfo {pages} {075103} (\bibinfo {year} {2011})},\ \Eprint {https://arxiv.org/abs/1008.4138} {arXiv:1008.4138 [cond-mat.str-el]} \BibitemShut {NoStop}%
\bibitem [{\citenamefont {Wen}(2013)}]{Wen:2013ppa}%
  \BibitemOpen
  \bibfield  {author} {\bibinfo {author} {\bibfnamefont {X.-G.}\ \bibnamefont {Wen}},\ }\bibfield  {title} {\bibinfo {title} {{A lattice non-perturbative definition of an SO(10) chiral gauge theory and its induced standard model}},\ }\href {https://doi.org/10.1088/0256-307X/30/11/111101} {\bibfield  {journal} {\bibinfo  {journal} {Chin. Phys. Lett.}\ }\textbf {\bibinfo {volume} {30}},\ \bibinfo {pages} {111101} (\bibinfo {year} {2013})},\ \Eprint {https://arxiv.org/abs/1305.1045} {arXiv:1305.1045 [hep-lat]} \BibitemShut {NoStop}%
\bibitem [{\citenamefont {Wang}\ and\ \citenamefont {Wen}(2023)}]{Wang:2013yta}%
  \BibitemOpen
  \bibfield  {author} {\bibinfo {author} {\bibfnamefont {J.}~\bibnamefont {Wang}}\ and\ \bibinfo {author} {\bibfnamefont {X.-G.}\ \bibnamefont {Wen}},\ }\bibfield  {title} {\bibinfo {title} {{Nonperturbative regularization of (1+1)-dimensional anomaly-free chiral fermions and bosons: On the equivalence of anomaly matching conditions and boundary gapping rules}},\ }\href {https://doi.org/10.1103/PhysRevB.107.014311} {\bibfield  {journal} {\bibinfo  {journal} {Phys. Rev. B}\ }\textbf {\bibinfo {volume} {107}},\ \bibinfo {pages} {014311} (\bibinfo {year} {2023})},\ \Eprint {https://arxiv.org/abs/1307.7480} {arXiv:1307.7480 [hep-lat]} \BibitemShut {NoStop}%
\bibitem [{\citenamefont {You}\ \emph {et~al.}(2014)\citenamefont {You}, \citenamefont {BenTov},\ and\ \citenamefont {Xu}}]{You:2014oaa}%
  \BibitemOpen
  \bibfield  {author} {\bibinfo {author} {\bibfnamefont {Y.}~\bibnamefont {You}}, \bibinfo {author} {\bibfnamefont {Y.}~\bibnamefont {BenTov}},\ and\ \bibinfo {author} {\bibfnamefont {C.}~\bibnamefont {Xu}},\ }\bibfield  {title} {\bibinfo {title} {{Interacting Topological Superconductors and possible Origin of $16n$ Chiral Fermions in the Standard Model}},\ }\href@noop {} {\  (\bibinfo {year} {2014})},\ \Eprint {https://arxiv.org/abs/1402.4151} {arXiv:1402.4151 [cond-mat.str-el]} \BibitemShut {NoStop}%
\bibitem [{\citenamefont {Ayyar}\ and\ \citenamefont {Chandrasekharan}(2015)}]{Ayyar:2014eua}%
  \BibitemOpen
  \bibfield  {author} {\bibinfo {author} {\bibfnamefont {V.}~\bibnamefont {Ayyar}}\ and\ \bibinfo {author} {\bibfnamefont {S.}~\bibnamefont {Chandrasekharan}},\ }\bibfield  {title} {\bibinfo {title} {{Massive fermions without fermion bilinear condensates}},\ }\href {https://doi.org/10.1103/PhysRevD.91.065035} {\bibfield  {journal} {\bibinfo  {journal} {Phys. Rev. D}\ }\textbf {\bibinfo {volume} {91}},\ \bibinfo {pages} {065035} (\bibinfo {year} {2015})},\ \Eprint {https://arxiv.org/abs/1410.6474} {arXiv:1410.6474 [hep-lat]} \BibitemShut {NoStop}%
\bibitem [{\citenamefont {BenTov}(2015)}]{BenTov:2014eea}%
  \BibitemOpen
  \bibfield  {author} {\bibinfo {author} {\bibfnamefont {Y.}~\bibnamefont {BenTov}},\ }\bibfield  {title} {\bibinfo {title} {{Fermion masses without symmetry breaking in two spacetime dimensions}},\ }\href {https://doi.org/10.1007/JHEP07(2015)034} {\bibfield  {journal} {\bibinfo  {journal} {JHEP}\ }\textbf {\bibinfo {volume} {07}},\ \bibinfo {pages} {034}},\ \Eprint {https://arxiv.org/abs/1412.0154} {arXiv:1412.0154 [cond-mat.str-el]} \BibitemShut {NoStop}%
\bibitem [{\citenamefont {You}\ and\ \citenamefont {Xu}(2015)}]{You:2014vea}%
  \BibitemOpen
  \bibfield  {author} {\bibinfo {author} {\bibfnamefont {Y.-Z.}\ \bibnamefont {You}}\ and\ \bibinfo {author} {\bibfnamefont {C.}~\bibnamefont {Xu}},\ }\bibfield  {title} {\bibinfo {title} {{Interacting Topological Insulator and Emergent Grand Unified Theory}},\ }\href {https://doi.org/10.1103/PhysRevB.91.125147} {\bibfield  {journal} {\bibinfo  {journal} {Phys. Rev. B}\ }\textbf {\bibinfo {volume} {91}},\ \bibinfo {pages} {125147} (\bibinfo {year} {2015})},\ \Eprint {https://arxiv.org/abs/1412.4784} {arXiv:1412.4784 [cond-mat.str-el]} \BibitemShut {NoStop}%
\bibitem [{\citenamefont {Catterall}(2016)}]{Catterall:2015zua}%
  \BibitemOpen
  \bibfield  {author} {\bibinfo {author} {\bibfnamefont {S.}~\bibnamefont {Catterall}},\ }\bibfield  {title} {\bibinfo {title} {{Fermion mass without symmetry breaking}},\ }\href {https://doi.org/10.1007/JHEP01(2016)121} {\bibfield  {journal} {\bibinfo  {journal} {JHEP}\ }\textbf {\bibinfo {volume} {01}},\ \bibinfo {pages} {121}},\ \Eprint {https://arxiv.org/abs/1510.04153} {arXiv:1510.04153 [hep-lat]} \BibitemShut {NoStop}%
\bibitem [{\citenamefont {Ayyar}\ and\ \citenamefont {Chandrasekharan}(2016{\natexlab{a}})}]{Ayyar:2015lrd}%
  \BibitemOpen
  \bibfield  {author} {\bibinfo {author} {\bibfnamefont {V.}~\bibnamefont {Ayyar}}\ and\ \bibinfo {author} {\bibfnamefont {S.}~\bibnamefont {Chandrasekharan}},\ }\bibfield  {title} {\bibinfo {title} {{Origin of fermion masses without spontaneous symmetry breaking}},\ }\href {https://doi.org/10.1103/PhysRevD.93.081701} {\bibfield  {journal} {\bibinfo  {journal} {Phys. Rev. D}\ }\textbf {\bibinfo {volume} {93}},\ \bibinfo {pages} {081701} (\bibinfo {year} {2016}{\natexlab{a}})},\ \Eprint {https://arxiv.org/abs/1511.09071} {arXiv:1511.09071 [hep-lat]} \BibitemShut {NoStop}%
\bibitem [{\citenamefont {He}\ \emph {et~al.}(2016)\citenamefont {He}, \citenamefont {Wu}, \citenamefont {You}, \citenamefont {Xu}, \citenamefont {Meng},\ and\ \citenamefont {Lu}}]{He:2016sbs}%
  \BibitemOpen
  \bibfield  {author} {\bibinfo {author} {\bibfnamefont {Y.-Y.}\ \bibnamefont {He}}, \bibinfo {author} {\bibfnamefont {H.-Q.}\ \bibnamefont {Wu}}, \bibinfo {author} {\bibfnamefont {Y.-Z.}\ \bibnamefont {You}}, \bibinfo {author} {\bibfnamefont {C.}~\bibnamefont {Xu}}, \bibinfo {author} {\bibfnamefont {Z.~Y.}\ \bibnamefont {Meng}},\ and\ \bibinfo {author} {\bibfnamefont {Z.-Y.}\ \bibnamefont {Lu}},\ }\bibfield  {title} {\bibinfo {title} {{Quantum critical point of Dirac fermion mass generation without spontaneous symmetry breaking}},\ }\href {https://doi.org/10.1103/PhysRevB.94.241111} {\bibfield  {journal} {\bibinfo  {journal} {Phys. Rev. B}\ }\textbf {\bibinfo {volume} {94}},\ \bibinfo {pages} {241111} (\bibinfo {year} {2016})},\ \Eprint {https://arxiv.org/abs/1603.08376} {arXiv:1603.08376 [cond-mat.str-el]} \BibitemShut {NoStop}%
\bibitem [{\citenamefont {Ayyar}\ and\ \citenamefont {Chandrasekharan}(2016{\natexlab{b}})}]{Ayyar:2016lxq}%
  \BibitemOpen
  \bibfield  {author} {\bibinfo {author} {\bibfnamefont {V.}~\bibnamefont {Ayyar}}\ and\ \bibinfo {author} {\bibfnamefont {S.}~\bibnamefont {Chandrasekharan}},\ }\bibfield  {title} {\bibinfo {title} {{Fermion masses through four-fermion condensates}},\ }\href {https://doi.org/10.1007/JHEP10(2016)058} {\bibfield  {journal} {\bibinfo  {journal} {JHEP}\ }\textbf {\bibinfo {volume} {10}},\ \bibinfo {pages} {058}},\ \Eprint {https://arxiv.org/abs/1606.06312} {arXiv:1606.06312 [hep-lat]} \BibitemShut {NoStop}%
\bibitem [{\citenamefont {You}\ \emph {et~al.}(2018)\citenamefont {You}, \citenamefont {He}, \citenamefont {Xu},\ and\ \citenamefont {Vishwanath}}]{You:2017ltx}%
  \BibitemOpen
  \bibfield  {author} {\bibinfo {author} {\bibfnamefont {Y.-Z.}\ \bibnamefont {You}}, \bibinfo {author} {\bibfnamefont {Y.-C.}\ \bibnamefont {He}}, \bibinfo {author} {\bibfnamefont {C.}~\bibnamefont {Xu}},\ and\ \bibinfo {author} {\bibfnamefont {A.}~\bibnamefont {Vishwanath}},\ }\bibfield  {title} {\bibinfo {title} {{Symmetric Fermion Mass Generation as Deconfined Quantum Criticality}},\ }\href {https://doi.org/10.1103/PhysRevX.8.011026} {\bibfield  {journal} {\bibinfo  {journal} {Phys. Rev. X}\ }\textbf {\bibinfo {volume} {8}},\ \bibinfo {pages} {011026} (\bibinfo {year} {2018})},\ \Eprint {https://arxiv.org/abs/1705.09313} {arXiv:1705.09313 [cond-mat.str-el]} \BibitemShut {NoStop}%
\bibitem [{\citenamefont {Kikukawa}(2019)}]{Kikukawa:2017gvk}%
  \BibitemOpen
  \bibfield  {author} {\bibinfo {author} {\bibfnamefont {Y.}~\bibnamefont {Kikukawa}},\ }\bibfield  {title} {\bibinfo {title} {{Why is the mission impossible? -- Decoupling the mirror Ginsparg-Wilson fermions in the lattice models for two-dimensional abelian chiral gauge theories}},\ }\href {https://doi.org/10.1093/ptep/ptz055} {\bibfield  {journal} {\bibinfo  {journal} {PTEP}\ }\textbf {\bibinfo {volume} {2019}},\ \bibinfo {pages} {073B02} (\bibinfo {year} {2019})},\ \Eprint {https://arxiv.org/abs/1710.11101} {arXiv:1710.11101 [hep-lat]} \BibitemShut {NoStop}%
\bibitem [{\citenamefont {Wang}\ and\ \citenamefont {Wen}(2020)}]{Wang:2018jkc}%
  \BibitemOpen
  \bibfield  {author} {\bibinfo {author} {\bibfnamefont {J.}~\bibnamefont {Wang}}\ and\ \bibinfo {author} {\bibfnamefont {X.-G.}\ \bibnamefont {Wen}},\ }\bibfield  {title} {\bibinfo {title} {{Nonperturbative definition of the standard models}},\ }\href {https://doi.org/10.1103/PhysRevResearch.2.023356} {\bibfield  {journal} {\bibinfo  {journal} {Phys. Rev. Res.}\ }\textbf {\bibinfo {volume} {2}},\ \bibinfo {pages} {023356} (\bibinfo {year} {2020})},\ \Eprint {https://arxiv.org/abs/1809.11171} {arXiv:1809.11171 [hep-th]} \BibitemShut {NoStop}%
\bibitem [{\citenamefont {Razamat}\ and\ \citenamefont {Tong}(2021)}]{Razamat:2020kyf}%
  \BibitemOpen
  \bibfield  {author} {\bibinfo {author} {\bibfnamefont {S.~S.}\ \bibnamefont {Razamat}}\ and\ \bibinfo {author} {\bibfnamefont {D.}~\bibnamefont {Tong}},\ }\bibfield  {title} {\bibinfo {title} {{Gapped Chiral Fermions}},\ }\href {https://doi.org/10.1103/PhysRevX.11.011063} {\bibfield  {journal} {\bibinfo  {journal} {Phys. Rev. X}\ }\textbf {\bibinfo {volume} {11}},\ \bibinfo {pages} {011063} (\bibinfo {year} {2021})},\ \Eprint {https://arxiv.org/abs/2009.05037} {arXiv:2009.05037 [hep-th]} \BibitemShut {NoStop}%
\bibitem [{\citenamefont {Tong}(2022)}]{Tong:2021phe}%
  \BibitemOpen
  \bibfield  {author} {\bibinfo {author} {\bibfnamefont {D.}~\bibnamefont {Tong}},\ }\bibfield  {title} {\bibinfo {title} {{Comments on symmetric mass generation in 2d and 4d}},\ }\href {https://doi.org/10.1007/JHEP07(2022)001} {\bibfield  {journal} {\bibinfo  {journal} {JHEP}\ }\textbf {\bibinfo {volume} {07}},\ \bibinfo {pages} {001}},\ \Eprint {https://arxiv.org/abs/2104.03997} {arXiv:2104.03997 [hep-th]} \BibitemShut {NoStop}%
\bibitem [{\citenamefont {Butt}\ \emph {et~al.}(2021)\citenamefont {Butt}, \citenamefont {Catterall},\ and\ \citenamefont {Toga}}]{Butt:2021koj}%
  \BibitemOpen
  \bibfield  {author} {\bibinfo {author} {\bibfnamefont {N.}~\bibnamefont {Butt}}, \bibinfo {author} {\bibfnamefont {S.}~\bibnamefont {Catterall}},\ and\ \bibinfo {author} {\bibfnamefont {G.~C.}\ \bibnamefont {Toga}},\ }\bibfield  {title} {\bibinfo {title} {{Symmetric Mass Generation in Lattice Gauge Theory}},\ }\href {https://doi.org/10.3390/sym13122276} {\bibfield  {journal} {\bibinfo  {journal} {Symmetry}\ }\textbf {\bibinfo {volume} {13}},\ \bibinfo {pages} {2276} (\bibinfo {year} {2021})},\ \Eprint {https://arxiv.org/abs/2111.01001} {arXiv:2111.01001 [hep-lat]} \BibitemShut {NoStop}%
\bibitem [{\citenamefont {Zeng}\ \emph {et~al.}(2022)\citenamefont {Zeng}, \citenamefont {Zhu}, \citenamefont {Wang},\ and\ \citenamefont {You}}]{Zeng:2022grc}%
  \BibitemOpen
  \bibfield  {author} {\bibinfo {author} {\bibfnamefont {M.}~\bibnamefont {Zeng}}, \bibinfo {author} {\bibfnamefont {Z.}~\bibnamefont {Zhu}}, \bibinfo {author} {\bibfnamefont {J.}~\bibnamefont {Wang}},\ and\ \bibinfo {author} {\bibfnamefont {Y.-Z.}\ \bibnamefont {You}},\ }\bibfield  {title} {\bibinfo {title} {{Symmetric Mass Generation in the 1+1 Dimensional Chiral Fermion 3-4-5-0 Model}},\ }\href {https://doi.org/10.1103/PhysRevLett.128.185301} {\bibfield  {journal} {\bibinfo  {journal} {Phys. Rev. Lett.}\ }\textbf {\bibinfo {volume} {128}},\ \bibinfo {pages} {185301} (\bibinfo {year} {2022})},\ \Eprint {https://arxiv.org/abs/2202.12355} {arXiv:2202.12355 [cond-mat.str-el]} \BibitemShut {NoStop}%
\bibitem [{\citenamefont {Wang}\ and\ \citenamefont {You}(2022)}]{Wang:2022ucy}%
  \BibitemOpen
  \bibfield  {author} {\bibinfo {author} {\bibfnamefont {J.}~\bibnamefont {Wang}}\ and\ \bibinfo {author} {\bibfnamefont {Y.-Z.}\ \bibnamefont {You}},\ }\bibfield  {title} {\bibinfo {title} {{Symmetric Mass Generation}},\ }\href {https://doi.org/10.3390/sym14071475} {\bibfield  {journal} {\bibinfo  {journal} {Symmetry}\ }\textbf {\bibinfo {volume} {14}},\ \bibinfo {pages} {1475} (\bibinfo {year} {2022})},\ \Eprint {https://arxiv.org/abs/2204.14271} {arXiv:2204.14271 [cond-mat.str-el]} \BibitemShut {NoStop}%
\bibitem [{\citenamefont {Liu}\ \emph {et~al.}(2024)\citenamefont {Liu} \emph {et~al.}}]{Liu:2023msa}%
  \BibitemOpen
  \bibfield  {author} {\bibinfo {author} {\bibfnamefont {Z.~H.}\ \bibnamefont {Liu}} \emph {et~al.},\ }\bibfield  {title} {\bibinfo {title} {{Disorder Operator and R{\'e}nyi Entanglement Entropy of Symmetric Mass Generation}},\ }\href {https://doi.org/10.1103/PhysRevLett.132.156503} {\bibfield  {journal} {\bibinfo  {journal} {Phys. Rev. Lett.}\ }\textbf {\bibinfo {volume} {132}},\ \bibinfo {pages} {156503} (\bibinfo {year} {2024})},\ \Eprint {https://arxiv.org/abs/2308.07380} {arXiv:2308.07380 [cond-mat.str-el]} \BibitemShut {NoStop}%
\bibitem [{\citenamefont {Butt}\ \emph {et~al.}(2025)\citenamefont {Butt}, \citenamefont {Catterall},\ and\ \citenamefont {Hasenfratz}}]{Butt:2024kxi}%
  \BibitemOpen
  \bibfield  {author} {\bibinfo {author} {\bibfnamefont {N.}~\bibnamefont {Butt}}, \bibinfo {author} {\bibfnamefont {S.}~\bibnamefont {Catterall}},\ and\ \bibinfo {author} {\bibfnamefont {A.}~\bibnamefont {Hasenfratz}},\ }\bibfield  {title} {\bibinfo {title} {{Symmetric Mass Generation with Four SU(2) Doublet Fermions}},\ }\href {https://doi.org/10.1103/PhysRevLett.134.031602} {\bibfield  {journal} {\bibinfo  {journal} {Phys. Rev. Lett.}\ }\textbf {\bibinfo {volume} {134}},\ \bibinfo {pages} {031602} (\bibinfo {year} {2025})},\ \Eprint {https://arxiv.org/abs/2409.02062} {arXiv:2409.02062 [hep-lat]} \BibitemShut {NoStop}%
\bibitem [{\citenamefont {Golterman}\ and\ \citenamefont {Shamir}(2026)}]{Golterman:2025boq}%
  \BibitemOpen
  \bibfield  {author} {\bibinfo {author} {\bibfnamefont {M.}~\bibnamefont {Golterman}}\ and\ \bibinfo {author} {\bibfnamefont {Y.}~\bibnamefont {Shamir}},\ }\bibfield  {title} {\bibinfo {title} {{Constraints on the symmetric mass generation paradigm for lattice chiral gauge theories}},\ }\href {https://doi.org/10.1103/qd1t-wzy1} {\bibfield  {journal} {\bibinfo  {journal} {Phys. Rev. D}\ }\textbf {\bibinfo {volume} {113}},\ \bibinfo {pages} {014503} (\bibinfo {year} {2026})},\ \Eprint {https://arxiv.org/abs/2505.20436} {arXiv:2505.20436 [hep-lat]} \BibitemShut {NoStop}%
\bibitem [{\citenamefont {Maiti}\ \emph {et~al.}(2026)\citenamefont {Maiti}, \citenamefont {Banerjee}, \citenamefont {Chandrasekharan},\ and\ \citenamefont {Marinkovic}}]{Maiti:2026log}%
  \BibitemOpen
  \bibfield  {author} {\bibinfo {author} {\bibfnamefont {S.}~\bibnamefont {Maiti}}, \bibinfo {author} {\bibfnamefont {D.}~\bibnamefont {Banerjee}}, \bibinfo {author} {\bibfnamefont {S.}~\bibnamefont {Chandrasekharan}},\ and\ \bibinfo {author} {\bibfnamefont {M.~K.}\ \bibnamefont {Marinkovic}},\ }\bibfield  {title} {\bibinfo {title} {{Phase diagram of a lattice fermion model with symmetric mass generation}},\ }\href@noop {} {\  (\bibinfo {year} {2026})},\ \Eprint {https://arxiv.org/abs/2602.18360} {arXiv:2602.18360 [hep-lat]} \BibitemShut {NoStop}%
\bibitem [{\citenamefont {Li}\ \emph {et~al.}(2026)\citenamefont {Li}, \citenamefont {Yu}, \citenamefont {Li},\ and\ \citenamefont {Yin}}]{Li:2026gvn}%
  \BibitemOpen
  \bibfield  {author} {\bibinfo {author} {\bibfnamefont {Z.-X.}\ \bibnamefont {Li}}, \bibinfo {author} {\bibfnamefont {Y.-K.}\ \bibnamefont {Yu}}, \bibinfo {author} {\bibfnamefont {Z.-X.}\ \bibnamefont {Li}},\ and\ \bibinfo {author} {\bibfnamefont {S.}~\bibnamefont {Yin}},\ }\bibfield  {title} {\bibinfo {title} {{Symmetric Mass Generation Transition and its Nonequilibrium Critical Dynamics in a Bilayer Honeycomb Lattice Model}},\ }\href@noop {} {\  (\bibinfo {year} {2026})},\ \Eprint {https://arxiv.org/abs/2603.22736} {arXiv:2603.22736 [cond-mat.str-el]} \BibitemShut {NoStop}%
\bibitem [{\citenamefont {Hasenfratz}\ and\ \citenamefont {Xu}(2026)}]{Hasenfratz:2026orf}%
  \BibitemOpen
  \bibfield  {author} {\bibinfo {author} {\bibfnamefont {A.}~\bibnamefont {Hasenfratz}}\ and\ \bibinfo {author} {\bibfnamefont {C.}~\bibnamefont {Xu}},\ }\bibfield  {title} {\bibinfo {title} {{A Guide to Symmetric Mass Generation in Lattice-QCD}},\ }\href@noop {} {\  (\bibinfo {year} {2026})},\ \Eprint {https://arxiv.org/abs/2604.02424} {arXiv:2604.02424 [hep-lat]} \BibitemShut {NoStop}%
\bibitem [{\citenamefont {Butt}\ \emph {et~al.}(2026)\citenamefont {Butt}, \citenamefont {Catterall}, \citenamefont {Hartshaw},\ and\ \citenamefont {Hasenfratz}}]{butt2026searchingsymmetricmassgeneration}%
  \BibitemOpen
  \bibfield  {author} {\bibinfo {author} {\bibfnamefont {N.}~\bibnamefont {Butt}}, \bibinfo {author} {\bibfnamefont {S.}~\bibnamefont {Catterall}}, \bibinfo {author} {\bibfnamefont {G.}~\bibnamefont {Hartshaw}},\ and\ \bibinfo {author} {\bibfnamefont {A.}~\bibnamefont {Hasenfratz}},\ }\href {https://arxiv.org/abs/2608.18239} {\bibinfo {title} {Searching for symmetric mass generation with staggered fermions in four dimensions}} (\bibinfo {year} {2026}),\ \Eprint {https://arxiv.org/abs/2608.18239} {arXiv:2608.18239 [hep-lat]} \BibitemShut {NoStop}%
\bibitem [{\citenamefont {'t~Hooft}(1980)}]{tHooft:1979rat}%
  \BibitemOpen
  \bibfield  {author} {\bibinfo {author} {\bibfnamefont {G.}~\bibnamefont {'t~Hooft}},\ }\bibfield  {title} {\bibinfo {title} {{Naturalness, chiral symmetry, and spontaneous chiral symmetry breaking}},\ }\href {https://doi.org/10.1007/978-1-4684-7571-5_9} {\bibfield  {journal} {\bibinfo  {journal} {NATO Sci. Ser. B}\ }\textbf {\bibinfo {volume} {59}},\ \bibinfo {pages} {135} (\bibinfo {year} {1980})}\BibitemShut {NoStop}%
\bibitem [{\citenamefont {Tong}\ and\ \citenamefont {Turner}(2020)}]{Tong:2019bbk}%
  \BibitemOpen
  \bibfield  {author} {\bibinfo {author} {\bibfnamefont {D.}~\bibnamefont {Tong}}\ and\ \bibinfo {author} {\bibfnamefont {C.}~\bibnamefont {Turner}},\ }\bibfield  {title} {\bibinfo {title} {{Notes on 8 Majorana Fermions}},\ }\href {https://doi.org/10.21468/SciPostPhysLectNotes.14} {\bibfield  {journal} {\bibinfo  {journal} {SciPost Phys. Lect. Notes}\ }\textbf {\bibinfo {volume} {14}},\ \bibinfo {pages} {1} (\bibinfo {year} {2020})},\ \Eprint {https://arxiv.org/abs/1906.07199} {arXiv:1906.07199 [hep-th]} \BibitemShut {NoStop}%
\bibitem [{\citenamefont {Shiozaki}\ \emph {et~al.}(2018)\citenamefont {Shiozaki}, \citenamefont {Shapourian}, \citenamefont {Gomi},\ and\ \citenamefont {Ryu}}]{Shiozaki:2017ive}%
  \BibitemOpen
  \bibfield  {author} {\bibinfo {author} {\bibfnamefont {K.}~\bibnamefont {Shiozaki}}, \bibinfo {author} {\bibfnamefont {H.}~\bibnamefont {Shapourian}}, \bibinfo {author} {\bibfnamefont {K.}~\bibnamefont {Gomi}},\ and\ \bibinfo {author} {\bibfnamefont {S.}~\bibnamefont {Ryu}},\ }\bibfield  {title} {\bibinfo {title} {{Many-body topological invariants for fermionic short-range entangled topological phases protected by antiunitary symmetries}},\ }\href {https://doi.org/10.1103/PhysRevB.98.035151} {\bibfield  {journal} {\bibinfo  {journal} {Phys. Rev. B}\ }\textbf {\bibinfo {volume} {98}},\ \bibinfo {pages} {035151} (\bibinfo {year} {2018})},\ \Eprint {https://arxiv.org/abs/1710.01886} {arXiv:1710.01886 [cond-mat.str-el]} \BibitemShut {NoStop}%
\bibitem [{\citenamefont {Araki}\ \emph {et~al.}(2026{\natexlab{a}})\citenamefont {Araki}, \citenamefont {Fukaya}, \citenamefont {Onogi},\ and\ \citenamefont {Yamaguchi}}]{Araki:2025xly}%
  \BibitemOpen
  \bibfield  {author} {\bibinfo {author} {\bibfnamefont {S.}~\bibnamefont {Araki}}, \bibinfo {author} {\bibfnamefont {H.}~\bibnamefont {Fukaya}}, \bibinfo {author} {\bibfnamefont {T.}~\bibnamefont {Onogi}},\ and\ \bibinfo {author} {\bibfnamefont {S.}~\bibnamefont {Yamaguchi}},\ }\bibfield  {title} {\bibinfo {title} {{Arf-Brown-Kervaire invariant on a lattice}},\ }\href {https://doi.org/10.1103/wyp3-bptq} {\bibfield  {journal} {\bibinfo  {journal} {Phys. Rev. D}\ }\textbf {\bibinfo {volume} {114}},\ \bibinfo {pages} {014503} (\bibinfo {year} {2026}{\natexlab{a}})},\ \Eprint {https://arxiv.org/abs/2512.11424} {arXiv:2512.11424 [hep-lat]} \BibitemShut {NoStop}%
\bibitem [{\citenamefont {Nagai}\ and\ \citenamefont {Tomiya}(2024)}]{Nagai:2024yaf}%
  \BibitemOpen
  \bibfield  {author} {\bibinfo {author} {\bibfnamefont {Y.}~\bibnamefont {Nagai}}\ and\ \bibinfo {author} {\bibfnamefont {A.}~\bibnamefont {Tomiya}},\ }\bibfield  {title} {\bibinfo {title} {{JuliaQCD: Portable lattice QCD package in Julia language}},\ }\href@noop {} {\  (\bibinfo {year} {2024})},\ \Eprint {https://arxiv.org/abs/2409.03030} {arXiv:2409.03030 [hep-lat]} \BibitemShut {NoStop}%
\bibitem [{\citenamefont {Araki}\ \emph {et~al.}(2026{\natexlab{b}})\citenamefont {Araki}, \citenamefont {Fukaya}, \citenamefont {Onogi},\ and\ \citenamefont {Yamaguchi}}]{Araki:2026prep}%
  \BibitemOpen
  \bibfield  {author} {\bibinfo {author} {\bibfnamefont {S.}~\bibnamefont {Araki}}, \bibinfo {author} {\bibfnamefont {H.}~\bibnamefont {Fukaya}}, \bibinfo {author} {\bibfnamefont {T.}~\bibnamefont {Onogi}},\ and\ \bibinfo {author} {\bibfnamefont {S.}~\bibnamefont {Yamaguchi}},\ }\bibfield  {title} {\bibinfo {title} {{Correlators of domain-wall fermions with symmetric mass generation}},\ }\href@noop {} {\  (\bibinfo {year} {2026}{\natexlab{b}})},\ \bibinfo {note} {in preparation}\BibitemShut {NoStop}%
\bibitem [{\citenamefont {Hernandez}\ \emph {et~al.}(1999)\citenamefont {Hernandez}, \citenamefont {Jansen},\ and\ \citenamefont {Luscher}}]{HernandezJansenLuscher}%
  \BibitemOpen
  \bibfield  {author} {\bibinfo {author} {\bibfnamefont {P.}~\bibnamefont {Hernandez}}, \bibinfo {author} {\bibfnamefont {K.}~\bibnamefont {Jansen}},\ and\ \bibinfo {author} {\bibfnamefont {M.}~\bibnamefont {Luscher}},\ }\bibfield  {title} {\bibinfo {title} {{Locality properties of Neuberger's lattice Dirac operator}},\ }\href {https://doi.org/10.1016/S0550-3213(99)00213-8} {\bibfield  {journal} {\bibinfo  {journal} {Nucl. Phys. B}\ }\textbf {\bibinfo {volume} {552}},\ \bibinfo {pages} {363} (\bibinfo {year} {1999})},\ \Eprint {https://arxiv.org/abs/hep-lat/9808010} {arXiv:hep-lat/9808010} \BibitemShut {NoStop}%
\end{thebibliography}%


\begin{thebibliography}{9}
\bibitem{HernandezJansenLuscher}
P.~Hern\'andez, K.~Jansen, and M.~L\"uscher,
\emph{Locality properties of Neuberger's lattice Dirac operator},
Nucl. Phys. B \textbf{552} (1999) 363--378,
arXiv:hep-lat/9808010.
\end{thebibliography}

\end{document}